\documentclass[journal,technote]{IEEEtran}

\usepackage{amssymb,amsmath,amsthm,amsfonts}
\usepackage{bm}

\newtheorem{theorem}{Theorem}
\newtheorem{definition}{Definition}
\newtheorem{lemma}{Lemma}
\newtheorem{proposition}{Proposition}
\newtheorem{corollary}{Corollary}

\def\diag{\mathop{\sf diag}}
\def\blkdiag{\mathop{\sf blkdiag}}

\def\T{\intercal}

\usepackage{xcolor}

\usepackage[numbers,sort&compress,square]{natbib}

   \usepackage[pdftex]{graphicx}
   \DeclareGraphicsExtensions{.pdf,.jpeg,.png}

\begin{document}
\title{
Qualitative stability of nonlinear networked systems
}

\author{Marco Tulio Angulo
        and~Jean-Jacques~Slotine
\thanks{M.T. Angulo (darkbyte@gmail.com) is with the Center for Complex Network Research (CCNR), Northeastern University; and with
Channing Division of Network Medicine, Brigham and Women's Hospital, and Harvard Medical School, Boston MA 02115, USA.
J.-J. Slotine (jjs@mit.edu) is with the Nonlinear Systems Laboratory, the Department of Mechanical Engineering and the Department of Brain and Cognitive Sciences, MIT, Cambridge, Massachusetts 02139, USA.}
}

\maketitle

\begin{abstract}

In many large systems, such as those encountered in biology or economics, the dynamics are nonlinear and are only known very coarsely. It is often the case,
however, that the signs (excitation or inhibition) of individual interactions are known. This paper extends to nonlinear systems  the classical criteria of linear sign stability
introduced in the 70's, yielding simple sufficient conditions to determine stability using only the sign patterns of the interactions.
\end{abstract}


\IEEEpeerreviewmaketitle


\section{Introduction}

It is challenging to determine the stability of large complex systems ---such as gene regulation networks, ecological systems and economic markets--- 
 not only because of their size, but because in practice we only know their dynamics very coarsely. 
Our  knowledge of these large complex systems  is  often limited to the interaction network  between the system's components (nodes) and the sign-pattern of the network  \cite{Barabasi:15}.
In gene regulation systems,  for instance, we know the complete list of genes  ---the  nodes of the network--- and the signs of their interactions ---positive interactions are activations and negative ones inhibitions  \cite{sontag2004some,kholodenko2002untangling}.
Yet, since the edge-weights and dynamic model associated to these networks are  unknown, this  provides ``qualitative'' information of the system only.

The classical \emph{sign-stability criterion} introduced in the 70's opened the door to study the qualitative stability of linear systems, requiring to know only the 
sign-pattern of their interconnection network in order to conclude stability  \cite{maybe1969qualitative,jeffries1987qualitative}.
Its impact was profound in diverse fields including ecology \cite{may1973qualitative}, economy \cite{quirk1965qualitative}, and more recently  control \cite{devarakonda2010engineering} and network science \cite{Angulo:15}. 
%
%
However,  the classical sign-stability criterion cannot be applied to nonlinear phenomena such as oscillations, bistability, chaos and so on that often appear in engineering, biological and economic systems
\cite{strogatz2014nonlinear}.
For nonlinear systems, stability is preserved under cascade (i.e., series) interconnections under quite general conditions  \cite{mezic2004coupled, russo2011graphical}.
Small-gain conditions  have also been proposed to understand the stability of nonlinear systems with particular interconnection networks \cite{sontag2006passivity,arcak2011diagonal}, 
but these conditions are not completely qualitative  as we need to establish the small-gain property of the nodal dynamics. 
Thus, a qualitative stability criterion  for nonlinear systems beyond the ``cascade of stable systems is stable'' is lacking.

In this note, we use  contraction theory \cite{lohmiller1998contraction} in order to extend the sign-stability criterion to nonlinear systems.  
Remarkably, almost the same conditions as in the linear sign-stability criterion  imply stability in nonlinear systems.
Indeed, the conditions for sign-stability of linear and nonlinear systems are identical if the ``asymmetries'' of the system are constant.
Otherwise it becomes necessary that the intrinsic stability of isolated nodes  is sufficiently strong.
%
%
The rest of the note is organized as follows. 
Section \ref{statement-mainresults} presents the problem statement and our main results, which are proved in Section III. Section IV extends our nonlinear sign-stability criterion to delayed interconnections and modules (i.e., sets of nodes), and presents a  discussion on the conditions for stability. Section V contains some concluding remarks.


\section{Problem statement and main results}
\label{statement-mainresults}

Consider a system composed of $n$  nodes where the scalar $x_i(t)$ denotes the state of node $i$ at time $t$. 
The state of a node may represent the abundance of certain specie in an ecological system, or the expression level of some gene in a gene regulatory network, and so on.
Suppose we are given a directed network $\mathcal G =({\sf V}, {\sf E})$ associated to the system with vertex set ${\sf V}=\{{\sf x}_1, \cdots, {\sf x}_n \}$  and  edge set ${\sf E}$. An edge  ${\sf x}_j \rightarrow {\sf x}_i$ in the network corresponds to a \emph{direct} influence of  $x_j$ on $x_i$.
We are also given the ``sign'' of each edge in the network, associating positive sign to ``activation'' and negative sign to ``repression''.

Our objective is to characterize sufficient conditions for the stability of the system from the knowledge of $\mathcal G$.
To address this problem, let us assume that the system dynamics satisfy the differential equations
\begin{equation}
\label{system}
\dot x_i = f_i (x_1, \cdots, x_n, t), \qquad i=1,\cdots, n, 
\end{equation}
for some smooth functions $f_i(x,t): \mathbb R^{n} \times \mathbb R_+ \rightarrow \mathbb R$, $x =(x_1, \cdots, x_n)^\T$. We might not know these functions exactly, but we assume we know some of their structure as precised later on.

Under model \eqref{system}, the edge ${\sf x}_j \rightarrow {\sf x}_i$ exists in $\mathcal G$ if and only if the function  $a_{ij}(x,t)=\partial f_i(x,t) / \partial x_j$ is not identically zero. Similarly, this edge is positive if $a_{ij}(x,t)>0$ and negative if $a_{ij}(x,t)<0$  \cite{kholodenko2002untangling}.
In general, the sign of $a_{ij}(x,t)$ changes with time but this will not introduce any conceptual difficulty to our analysis.
In particular, when the functions $f_i(x,t)$ are time invariant and monotone in $x$ ---as in most models of gene regulation--- the sign of $a_{ij}(x,t)$ remains constant.

We  use  the following notion of stability \cite{lohmiller1998contraction}:

\begin{definition} System \eqref{system} is contracting if for any two initial conditions their corresponding trajectories converge exponentially fast towards each other. 
\end{definition}

A contracting system forgets exponentially fast its initial condition converging towards a unique trajectory.  In contrast to Lyapunov stability, such limiting trajectory is not necessarily constant (i.e., an equilibrium).
Historically, basic convergence results on contracting systems can be traced back to \cite{lewis1949metric} using Finsler metrics, and also to  \cite{hartman1961stability,demidovich1962dissipativity} and \cite{desoer1972measure}.

Our main result is:

\begin{theorem} System \eqref{system} is contracting provided that:
\begin{itemize}
\item[(i)] Reciprocate interactions have opposite signs:
if  ${\sf x}_i \rightarrow {\sf x}_j$ and ${\sf x}_j \rightarrow {\sf x}_i$ are edges in the network, $ i \neq j$, then $a_{ji}(x,t) = - b_{ij} (x,t) a_{ij}(x,t)$ for some  function $b_{ij}(x,t) > 0$.
\item[(ii)] Self-loops are negative and strong enough to overcome the logarithmic derivative of the asymmetries: there exists functions $\alpha_i(x,t) >0$ such that  $a_{ii}(x,t) = - \alpha_i(x,t)$ and 
$$\alpha_{i}(x,t) > \frac{1}{2} \sum_{j \in \mathcal N_i} \left. \frac{\dot b_{ij}(x,t)}{ b_{ij}(x,t)} \right., \quad i=1, \cdots, n,$$
where  $\mathcal N_i$ is the set of feedback neighbors of node $i$.
\item[(iii)] $\mathcal G$ does not contain cycles of length $3$ or more.
\end{itemize}
\end{theorem}

\begin{corollary} Suppose that the asymmetries $b_{ij}>0$ are constant (i.e., reciprocate interactions have the same functional form). Then \eqref{system} is contracting if $a_{ii}(x,t)<0$ and $\mathcal G$ does not have cycles of length $3$ or more. 
\end{corollary}

We say that \eqref{system} is \emph{sign-stable} if it satisfies  the conditions of Theorem 1.  When these conditions hold in a region $\Omega \subseteq \mathbb R^n$ only, we say that \eqref{system} is sign-stable in $\Omega$.
The functions $b_{ij}(x,t)$ in condition (i) represent the \emph{asymmetry} between the edges involved in a negative feedback loop, Fig.1a. 
Corollary 1 implies that the sign-stability conditions for linear and nonlinear systems are identical if the asymmetries in the system are constant.
When the asymmetries are not constant, condition (ii) requires that the intrinsic stability of the nodes (represented by their contraction rates $\alpha_i(x,t)>0$) dominates their logarithmic derivative. 
This condition of sufficiently large contraction rates of the isolated nodes turns out to be necessary for contraction with time-varying asymmetries, and it can also be stated in terms of how ``fast'' the asymmetry needs to be (Section IV-B).
The  \emph{feedback neighbors} $\mathcal N_i$ of node  ${\sf x}_i$ is the set of nodes ${\sf x}_j$ such that the edges ${\sf x}_i \rightarrow {\sf x}_j$ and ${\sf x}_j \rightarrow {\sf x}_i$ exist.
A \emph{cycle} in $\mathcal G$ is a sequence of edges (excluding self loops) that starts and ends in the same node,  and its length is the number of edges it contains. 
For example $\{ {\sf x}_i \rightarrow {\sf x}_j,  {\sf x}_j \rightarrow  {\sf x}_k,  {\sf x}_k \rightarrow  {\sf x}_i \}$ with $i\neq j \neq k$ is a cycle of length 3. 
Condition (iii) is related to the fact that systems with cycles of length 3 require specific tuning of their edge-weights to be stable.
In the case of linear systems with $n=3$, this can be seen from the root-locus diagram: there are three branches and, by symmetry, at least one branch crosses to the right half-plane of the complex plane. 
Consequently, the stability depends on the edge-weights of the network and we cannot determine stability based on its signs only.
In Section IV we  extend Theorem 1 to time-delayed interconnections, and to interconnections of modules instead of nodes.
%
When the functions $f_i$ are linear and time invariant, our result reduces to the classical sign-stability criterion  \cite{maybe1969qualitative,jeffries1987qualitative}. 

\vspace*{0.1cm}

\noindent \emph{Example 1:} Applying Corollary 1, the system
$\dot x_1 = -x_1 - x_1 x_2,  \dot x_2 = x_1^2 - x_2 - x_2 x_3,   \dot x_3 = x_2^2 - x_3  $
is sign-stable in the region $\Omega =\{x \in \mathbb R^3 | x_2 > -1, x_3>-1\}$. Indeed, it satisfies condition (i) with $b_{12}=b_{23}=2$,
and condition (ii)  because $a_{11} = -1 - x_2 <0$, $a_{22}=  -1 - x_3 <0$, $a_{33} = -1$. Condition (iii) is satisfied because its associated network  does not have cycles of length 3 or more,  Fig.\ref{fig:schwarz-forms}b.   Furthermore, the origin $x \equiv 0$ is exponentially stable since it is a trajectory in $\Omega$.
\vspace*{0.1cm}

Our proof of  the main result is based on two observations. First,  the network of a sign-stable system is composed of cascades of what we call ``feedback chains'' (Proposition 2). 
Feedback chains are systems recursively  built using negative feedback interconnections under the constraint that they do not have cycles of length 3 or more.
Second,  feedback chains are contracting provided that their nodes, when isolated, have a large enough contraction rate to dominate the logarithmic derivative of the asymmetry of the system (Proposition 1).  
From these two observations, the proof of Theorem 1 follows directly, because the cascade interconnection of contracting systems remains contracting \cite{lohmiller1998contraction}.


\section{Proof of the main result}

\begin{figure*}[htbp]
\begin{center}
\includegraphics[width=18cm]{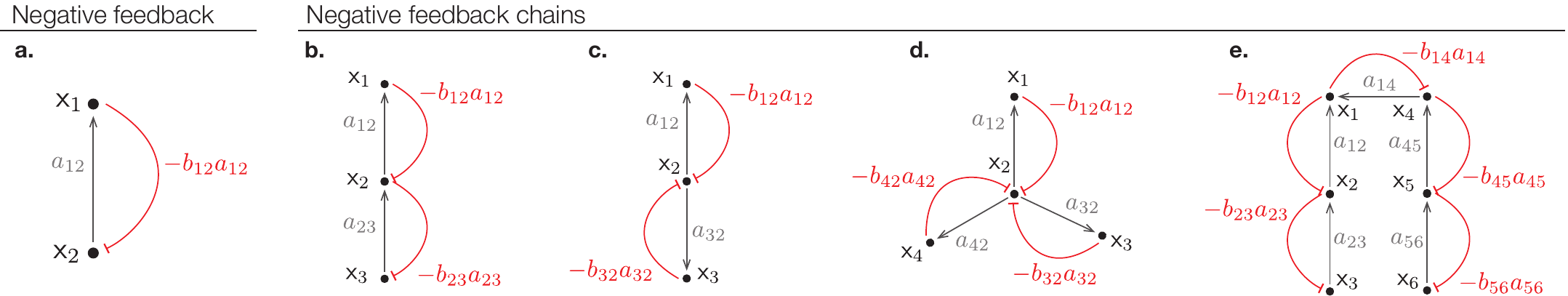}
\caption{{\bf  Two-node negative feedback interconnection, and some examples of negative feedback chains.} 
In all networks shown, the sign of  $a_{ij}$ may change with time because it is a function of $x(t)$ and $t$. Hence, the different edge representation for pairs of activation ``$\rightarrow$'' and inhibition ``$\dashv$'' indicates that $a_{ij}$ and $b_{ij} a_{ij}$ have opposite signs at each time instant.
All nodes have negative self-loops $a_{ii}= -\alpha_i  <0$, which are omitted when a network is displayed in order to improve readability.
}
\label{fig:schwarz-forms}
\vspace*{-0.5cm}
\end{center}
\end{figure*}

\subsection{Contraction theory, and its tools to prove stability.} Let $f= (f_1, \cdots, f_n)^\T$ allowing us to rewrite \eqref{system} as $\dot x = f(x, t)$. 
Denote its Jacobian  by $A(x,t) = \tfrac{\partial f(x,t)}{\partial x}$, and note its $(i,j)$ entry is the function $a_{ij}(x,t)$ used earlier to construct the network $\mathcal G$.
Hereafter, we often omit the arguments of the functions to improve readability unless they are relevant for the discussion.
Contraction theory uses the fact that exponential stability of the  \emph{differential system}  $\dot{\delta}_x(t) = A(x,t) \delta_x(t)$ implies that system \eqref{system} is contracting \cite{lohmiller1998contraction}.
In order to prove contraction, a necessary and sufficient condition is a symmetric positive-definite matrix $M(t, x)$ such that
\begin{equation} 
\label{F-eq}
L(A,M)=\dot M + A^\T M + M A \prec 0.
\end{equation}
The matrix $M$ introduces a metric in the differential coordinates, and we say that the system is \emph{diagonally contracting} if $M$ in \eqref{F-eq} can be taken diagonal.  
Changes of coordinates are also useful, as illustrated in Example 4.


\subsection{Feedback chains, and their diagonal contraction properties.}

The basic building block for our analysis will be  the  negative feedback interconnection for $n=2$ nodes shown in Fig.\ref{fig:schwarz-forms}a. Its corresponding Jacobian is
\begin{equation}
\label{system-2}
 A_{12}=  \begin{pmatrix} -\alpha_1 & a_{12} \\ -b_{12} a_{12}& - \alpha_2 \end{pmatrix},
\end{equation}
with $\alpha_i(x,t)>0$ and $b_{12}(x,t) >0$.  
The particular structure of \eqref{system-2} suggests using the diagonal metric $D_{12} = \diag \{ b_{12}, 1 \} $ to prove contraction:
$$L(A_{12},D_{12}) = \begin{pmatrix}-2 \alpha_1 b_{12}+ \dot b_{12} & 0 \\ 0 & -2 \alpha_2 \end{pmatrix}$$
which will be uniformly negative-definite if
\begin{equation*}
\label{twonode-stability}
\alpha_1 >  0.5 \, \, \dot b_{12} / b_{12} , \quad \alpha_2>0.
\end{equation*}
Hence a negative feedback interconnection between two nodes is contracting provided that condition (ii) holds. 
In this case not only the metric $D_{12}$ is diagonal, but also $L(A_{12}, D_{12})$ is diagonal.
This property will be instrumental  in order to extend this result to more general feedback interconnections, which we call ``feedback chains''.
For linear systems, the diagonal stability  of systems with the structure \eqref{system-2}  is consequence of the so-called Schwarz form
 \cite{shorten2014classical}.
Here contraction theory naturally extends this  property to nonlinear systems.

 The systems shown from Fig.\ref{fig:schwarz-forms}c to  Fig.\ref{fig:schwarz-forms}e are other examples of negative feedback interconnections. Their linear versions are no longer Schwarz forms of any order, so we call them \emph{feedback chains}.
Feedback chains are recursively built starting from a negative feedback interconnection between two nodes,  and adding new nodes by interconnecting each one of them to a single existing node using negative feedback (i.e.,  reciprocate interactions with $b_{ij} >0$).
By  their construction, feedback chains do not contain cycles of length 3 or more.

\vspace*{0.1cm}

\noindent \emph{Example 2:} Consider Fig.\ref{fig:schwarz-forms}c and let $A_{123}$ denote its  Jacobian
$$A_{123}= \begin{pmatrix} -\alpha_1 & a_{12} & 0 \\ -b_{12} a_{12} & -\alpha_2 & - b_{32}a_{32} \\ 0 & a_{32} & - \alpha_3\end{pmatrix}. $$
If condition (ii) holds, the first cycle  $\{x_1, x_2\}$ is contracting with metric $D_{12} = \diag \{ b_{12}, 1 \} $; similarly, the second one $\{x_2, x_3\}$ is contracting with metric $D_{23}=\diag \{1, b_{32} \}$.  
We  naturally combine both metrics $D_{123}=\diag\{ b_{12},  1, b_{32} \} $ obtaining
\begin{equation*}
\label{example2}
L(A_{123}, D_{123}) = \begin{pmatrix} L(A_{12}, D_{12}) & 0_{2 \times 1} \\ 0_{1 \times 2} & -2 \alpha_3 b_{32} + \dot b_{32}\end{pmatrix} 
\end{equation*}
proving that the system corresponding to Fig.\ref{fig:schwarz-forms}c is contracting provided that node ${\sf x}_3$ satisfies condition (ii).  Here $L(A_{123}, D_{123}) $ is again diagonal, and recursively depends on $L(A_{12}, D_{12})$. The system in Fig. \ref{fig:schwarz-forms}b. can be similarly analyzed  defining   $a_{32} = -b_{23} a_{23}$ and $b_{32}=b_{23}^{-1}$.

\noindent Now consider network Fig.\ref{fig:schwarz-forms}d  obtained by adding the node ${\sf x}_4$ to the feedback chain of Fig.\ref{fig:schwarz-forms}c. For  subsystem $A_{123}$  we  have the metric $D_{123}$, and for the subsystem $A_{24}$ we have the metric $D_{24}=\diag\{1, b_{42} \}$. Naturally combining both metrics $D_{1234}=\diag\{b_{12},1, b_{32} ,b_{42} \}$ we obtain %
$$L(A_{1234},D_{1234}) = \begin{pmatrix} L(A_{123}, D_{123}) & 0_{3 \times 1} \\ 0_{1 \times 3} & - 2 \alpha_4 b_{42} + \dot b_{42}\end{pmatrix}  $$
 which is again diagonal and proves that the system corresponding to network Fig.\ref{fig:schwarz-forms}d  is contracting if it satisfies condition (ii).

The above example show how to recursively build diagonal  metrics to prove contraction of feedback chains, illustrating the following result.

\begin{proposition} Provided that nodes satisfy condition (ii), a negative feedback chain is diagonally contracting.
\end{proposition}

\begin{proof} We prove the claim by induction on the dimension of the chain. Let $\Sigma_i$ denote a feedback chain of dimension $i$ and  $A_i$ its  Jacobian. 
We have shown that $\Sigma_2$ (i.e.,  two-node negative feedback) is diagonally contracting: there exists a diagonal metric  $D_2$ such that $L(A_2, D_2) \prec 0$. Furthermore, $L(A_2, D_2)$  is also diagonal. 

Now suppose that that $\Sigma_{i}$ is diagonally contracting and  $L(A_i, D_{i}) \prec 0$ is  diagonal. Let ${\sf x}_{i+1}$ be the new added node. By relabeling the nodes in $\Sigma_i$, we assume that ${\sf x}_{i+1}$ will be connected to ${\sf x}_i$ without loss of generality. 
Therefore the Jacobian of $\Sigma_{i+1}$ is
\begin{equation}
\label{G_i+1}
A_{i+1} = \begin{pmatrix} A_i & -b_{i+1 }a_{i+1} \bm e_i \\ a_{i+1} \bm e_i^\T & -\alpha_{i+1} \end{pmatrix} 
\end{equation}
where $\bm e_i = (0, \cdots, 0, 1) \in \mathbb R^i$. Here we have used $a_{i+1}$ and $b_{i+1}$ instead of $a_{i+1, i}$ and $b_{i+1, i}$ to simplify the notation.
As noted in Example 2, we may rewrite the Jacobian using $a_{i,i+1}=-b_{i+1, i}a_{i+1, i}$ and $b_{i, i+1}=b_{i+1,i}^{-1}$ so the expression \eqref{G_i+1} can be considered without loss of generality. 
Based on  $D_i \in \mathbb R^{i \times i}$, we build the new  metric
$$D_{i+1}= \begin{pmatrix}  D_i & 0 \\ 0 & d_{ii} b_{i+1} \end{pmatrix} $$ 
where $d_{ii}>0$ is the $i$-th diagonal element  of $D_i$. With this choice we obtain
$$D_{i+1} A_{i+1} = \begin{pmatrix} D_i A_i & - b_{i+1}a_{i+1} D_i \bm e_i \\ b_{i+1} a_{i+1} d_{ii} \bm e_i^\T & - d_{ii} b_{i+1} \alpha_{i+1} \end{pmatrix} .$$
Notice that $D_i \bm e_i = d_{ii} \bm e_i$, so we get
$$ L(A_{i+1}, D_{i+1}) = \begin{pmatrix}   L(A_i, D_i) & 0 \\ 0 & - 2 \alpha_{i+1} d_{ii} b_{i+1} + \frac{d}{dt}\{ d_{ii} b_{i+1} \} \end{pmatrix}. $$
From the induction hypothesis, we know that $L(A_i, D_i)$ is diagonal and uniformly negative definite. On the other hand, notice that  $d_{ii} = \prod_{j \in \bar {\mathcal  N}_i} b_j$ where  $\bar {\mathcal  N}_i$ is the set feedback neighbors of node ${\sf x}_i$ discarding ${\sf x}_{i+1}$. Using these two facts and the rule for the derivative of product of functions, $L(A_{i+1}, D_{i+1}) \prec 0$ provided that
\begin{eqnarray*}
- 2 \alpha_{i+1} b_{i+1} \left( \prod_{j \in \bar {\mathcal  N}_i} b_j  \right) +  \hspace*{3.5cm} \\
 +  b_{i+1} \left( \prod_{j \in \bar {\mathcal  N}_i} b_j \right) \left( \frac{\dot b_{i+1}}{b_{i+1}}+  \sum_{j \in \bar {\mathcal  N}_i} \frac{\dot b_j}{b_j} \right) < 0 
\end{eqnarray*}
which is  condition (ii) once it is divided by $\prod_{j \in  {\mathcal  N}_i} b_j$.
 \end{proof}

Despite it is possible to prove Proposition 1 without using recursion, the above proof is more natural from the point of view of a growing network, building larger metrics based on existing smaller ones. This recursive method will also be useful to extend the sign-stability criterion to delayed interconnections and modules (Section IV).


\subsection{The network structure of sign-stable systems.}

To conclude the proof of the main result,  we show that when $\mathcal G$ does not contain cycles of length $3$ or more ---that is, condition (iii) holds---  then the network is composed of cascades of feedback chains. 
We start showing that feedback chains are the minimal building blocks of networks with such constraint, and then we characterize how to interconnect feedback chains while keeping this constraint.

\begin{lemma}Adding a new edge to a feedback chain  produces a cycle of length 3 or more.
\end{lemma}
\begin{proof} By construction, a feedback chain is strongly connected (i.e., there is a direct path between any two nodes) and there is no cycle of length 3 or more. 
Now suppose we add a new edge ${\sf x}_j \rightarrow {\sf x}_i$.  
Yet, there is already a path in the network from ${\sf x}_i$ to ${\sf x}_j$ and since the edge is new, the length of this path should be $2$ at least. 
Hence the cycle obtained by the new edge and the already present path has length 3 or more.
\end{proof}

\begin{lemma} Two feedback chains can be interconnected without creating a cycle of length 3 or more only by (i)  a two-node feedback (so both feedback chains are merged into a larger one), or (ii) in cascade.
\end{lemma}
\begin{proof}
$(i)$
If we interconnect two chains by a feedback   we create a larger feedback chain, see Fig.\ref{fig:schwarz-forms3}a. Lemma 1 shows that adding another edge to this network creates a cycle of length 3 or more, Fig.\ref{fig:schwarz-forms3}b.
$(ii)$ Consider an edge going from chain $\Sigma_a$ to chain $\Sigma_b$ and suppose there is an another edge going from $\Sigma_b$ to $\Sigma_a$. This creates a cycle of length 3 or more, because each feedback chain is strongly connected  (see Fig.\ref{fig:schwarz-forms3}d). Consequently, edges can only go  from one feedback chain to the other, Fig.\ref{fig:schwarz-forms3}c.
\end{proof}

\begin{proposition}  Under conditions (i) and (iii), the network $\mathcal G$ of the system is a cascade of negative feedback chains.\end{proposition}

The proof of Proposition 2 is a direct consequence of Lemma 1 and 2. Figure 2a-2d illustrates these lemmas. 
Together with Proposition 1, Proposition 2  completes the proof of Theorem 1 because the cascade of contracting systems is  contracting \cite{lohmiller1998contraction}.
Corollary 1 follows from Theorem 1 because if $b_{ij}$ is constant then its logarithmic derivative is zero.

\begin{figure}[htbp]
\begin{center}
\includegraphics[width=7.5cm]{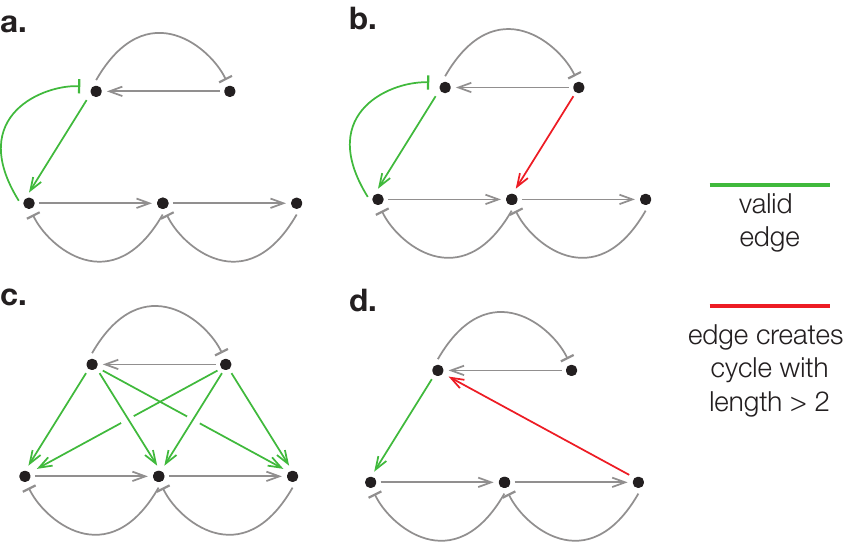}
\caption{{\bf Interconnection of feedback chains without creating cycles of length $> 2$.} {\bf a.} Interconnecting feedback chains using two-node feedback does not create cycles of length $>2$, merging together two chains into a larger one (Lemma 2-i). {\bf b.}  Adding any edge to a feedback chain always create cycles of length $>2$ (Lemma 1). {\bf c.} Interconnecting two feedback chains in cascade does not create cycles of length $>2$ (Lemma 2-ii). {\bf d.} If the interconnection between two feedback chains have different directions and are not feedback to the same node, they necessarily create a cycle with length $>2$ (Lemma 2-ii).
}
\vspace*{-0.5cm}
\label{fig:schwarz-forms3}
\end{center}
\end{figure}


\section{Discussion and extensions}

\subsection{Delayed interconnections.}

Delays in interconnections happen in natural and technological systems. Here we remark that for systems with uniformly bounded interactions,  there exists a critical threshold for the contraction rates of the nodes such that the system is contracting for \emph{any} value of delay. 


\vspace*{0.1cm}
\noindent \emph{Example 3:} Consider a two-node negative feedback interconnection with delays $T_1, T_2 \geq 0$, which corresponds to the following  differential system
\begin{equation*}
\begin{split}
\dot \delta_{x_1}(t) =& -\alpha_1 \delta_{x_1}(t) + a_{12} \delta_{x_2}(t- T_2), \\
\dot \delta_{x_2}(t) =& -b_{12} a_{12}  \delta_{x_1}(t-T_1)  - \alpha_2  \delta_{x_2}(t). 
\end{split}
\end{equation*}
Suppose there exits a constant $\Gamma \geq 0$ such that $|a_{12}(x,t)| b_{12}(x, t)  \leq \Gamma, \forall (x,t)$. 
Consider  $V(\delta_x) = \delta_x^\T D_{12} \delta_x$ with $D_{12}=\diag\{b_{12},1\}$ as previously chosen. In contrast to our previous analysis, the crossed terms do not cancel due to the delays
\begin{equation*}
\begin{split}
\dot V =& \delta_x^\T(t) L(A_{12},D_{12}) \delta_x(t) +
\\& 2 a_{12} b_{12} \left [ \delta_{x_1}(t) \delta_{x_2}(t-T_2)- \delta_{x_1}(t-T_1)\delta_{x_2}(t)    \right ] ,
\end{split}
\end{equation*}
but note that the first term of the above equation is still negative. Using the upper bound for the interactions $\Gamma$ and the inequality 
\begin{equation}
\label{schwarz}
2 | \delta_{x_i}(t) \delta_{x_j}(t-T_j)| \leq    \delta_{x_i}^2(t) +  \delta_{x_j}^2(t-T_j) 
\end{equation}
 we obtain $\dot V \leq \delta_x^\T(t) L(A_{12},D_{12}) \delta_x(t) +\Gamma [ \delta_{x_1}^2(t) + \delta_{x_2}^2(t-T_2) + \delta_{x_1}^2(t-T_1)+ \delta_{x_2}^2(t)     ]$.
Consider now the strictly positive function
$$V_{12}=  \int_{t-T_1}^t \delta_{x_1}^2(\sigma)  d\sigma + \int_{t-T_2}^t \delta_{x_2}^2(\sigma)  d\sigma$$
with time derivative
$ \dot V_{12} = \delta_{x_1}^2(t) - \delta_{x_1}^2(t-T_1) +\delta_{x_2}^2(t) - \delta_{x_2}^2(t- T_2)$.
When we combine these two functions $V_{tot}=V +   \Gamma V_{12}$ the delayed terms  cancel out:
$$\dot V_{tot} \leq - (2\alpha_1 b_{12} - \dot b_{12} - 2 \Gamma) \delta_{x_1}^2(t) - 2(\alpha_2 - \Gamma) \delta_{x_2}^2(t).$$
Consequently, if the contraction rates of the isolated nodes satisfy
\begin{equation}
\label{delay-2}
\alpha_1 > (0.5 \dot b_{12} + \Gamma)/ b_{12}, \quad \alpha_2 > \Gamma,
\end{equation}
then the system is contracting. In general, the bound $\Gamma$ may increase with the state. However, when the interconnection functions are independent of the state (e.g., time varying functions),  condition \eqref{delay-2} is sufficient to ensure contraction for \emph{any} value of the delays. Indeed, in the case of linear time-invariant systems, the characteristic equation of the system is
 $$ 1 + L(s)= 1+  \frac{a_{12} b_{12} e^{-s(T_1 + T_2)}}{(s+\alpha_1) (s+\alpha_2) } $$
 where $s \in \mathbb C$ is the Laplace variable. Observing that for any given $\Gamma \geq0$ there exists constant $\alpha_1, \alpha_2>0$ such that $|L(\imath \omega) | < 1$ for all $T_1, T_2 \geq 0$, the Nyquist criterion immediately implies that the system is stable for \emph{any} delay. 
\vspace*{0.1cm}

The method of the above example applies directly to the general case. We start with $V=\delta_x^\T D \delta_x$ with the metric $D$ constructed from Proposition 1. Next we apply  \eqref{schwarz} to replace the cross terms by sums of squares. Finally we add 
 $V_{i:k} = \sum_{j=i}^k \int_{t-T_j}^t \delta_{x_j}^2(\sigma) d \sigma$ to replace the delayed squared terms by squared terms without delays. 
 Note this is not restricted to the particular network structure of sign-stable systems.

\subsection{Time-varying asymmetries need large enough contraction rates.}

The linear sign-stability criterion \cite{maybe1969qualitative,jeffries1987qualitative} and the nonlinear criterion with constant asymmetries (Corollary 1) do not require the ``large enough'' contraction rate condition for $\alpha_i$ ---condition (ii) of Theorem 1--- which may be difficult to establish in practice.
In this sense, both criteria are purely qualitative (despite its necessary to establish linearity or  constant asymmetries in advance).
Nevertheless, even in the linear case, the condition of large enough contraction rates turns out to be necessary when the asymmetries are time-varying. 
This condition  can also be expressed in terms of how ``fast'' the asymmetries need to be.
We illustrate these two statements using the following example.

\vspace*{0.1cm}
\noindent \emph{Example 4:} Consider a linear time-varying instance of the two-node negative feedback interconnection \eqref{system-2}, using $\alpha_1 = \alpha_2 = \alpha=$const., $a_{12}=1$ and $b_{12}  = b(t)= 1 + 0.9 \sin (t) >0$.
Using numerical simulations,  we conclude this system is unstable for $\alpha \leq 0.05$ but stable for $\alpha \geq 0.09$, see Fig. \ref{fig:unstable}.
In other words, it is necessary a large enough contraction rate $\alpha>0$ for stability.

\noindent 
Next recall from \cite[Theorem 2]{aeyels1999exponential} that the linear system $\dot {\delta}_x(t) = A(\omega t) \delta_x(t)$ is globally exponentially stable $\forall \omega > \omega^*$ for some finite $\omega^* >0$ if there exists $T>0$ such that the ``average'' $\int_t^{t+T} A(\tau) d \tau$ is uniformly negative definite for all $t \geq 0$. 
In order to apply this result, we first rewrite our example in coordinates $\Theta(t) = \diag\{1, b(t)^{-1/2} \}$.
In these coordinates the Jacobian reads $F = \Theta A \Theta^{-1} + \dot \Theta \Theta^{-1}$, so we obtain
   $F(t) = \diag\{- 2 \alpha,- 2\alpha + \dot b(t)/ b(t)  \}$.
 Now we can compute  
 $$ \int_t^{t+T} F(\tau) d \tau = \begin{pmatrix} - 2 \alpha T & 0 \\ 0 & - 2 \alpha T + \int_t^{t+T} \frac{d}{dt} \ln (b) d \tau \end{pmatrix}.$$
The above matrix is negative definite in the identify metric if
$$-2 \alpha T + \ln \frac{b(t + \tau)}{b(t)} < 0 $$
which can be rewritten as $b(t+T ) < e^{2 \alpha T} b(t)$. 
For any $\alpha>0$ there is always $T>0$ satisfying this condition if  the function $b_{12}(t)$ is uniformly bounded from above and below.  Indeed, if $b_{12}(t) \geq \epsilon >0$ and $b_{12}(t) \leq \gamma$ we can choose any $T > (1/2 \alpha) \ln(\gamma/\epsilon)$.
This shows that for any $\alpha>0$, the system is contracting if  $\omega$ is large enough. In other words, if the asymmetry is bounded from above and below and  fast enough.
In our example with  $b(\omega t) = 1 + 0.9 \sin(\omega t)$ and $\alpha=0.05$, the system becomes stable by increasing the frequency to $\omega = 1.1$.  

\noindent On the other hand, Theorem 1 also implies that for any given $\alpha>0$ there exists $\omega_*>0$  small enough such that the system remains contracting for any $\omega < \omega_*$. Indeed, our example with  $\alpha=0.05$ becomes stable by decreasing the frequency of $b(\omega t)$ to $\omega =0.5$. 
\vspace*{0.1cm}

\begin{figure}[htbp]
\begin{center}
\includegraphics[width=9cm]{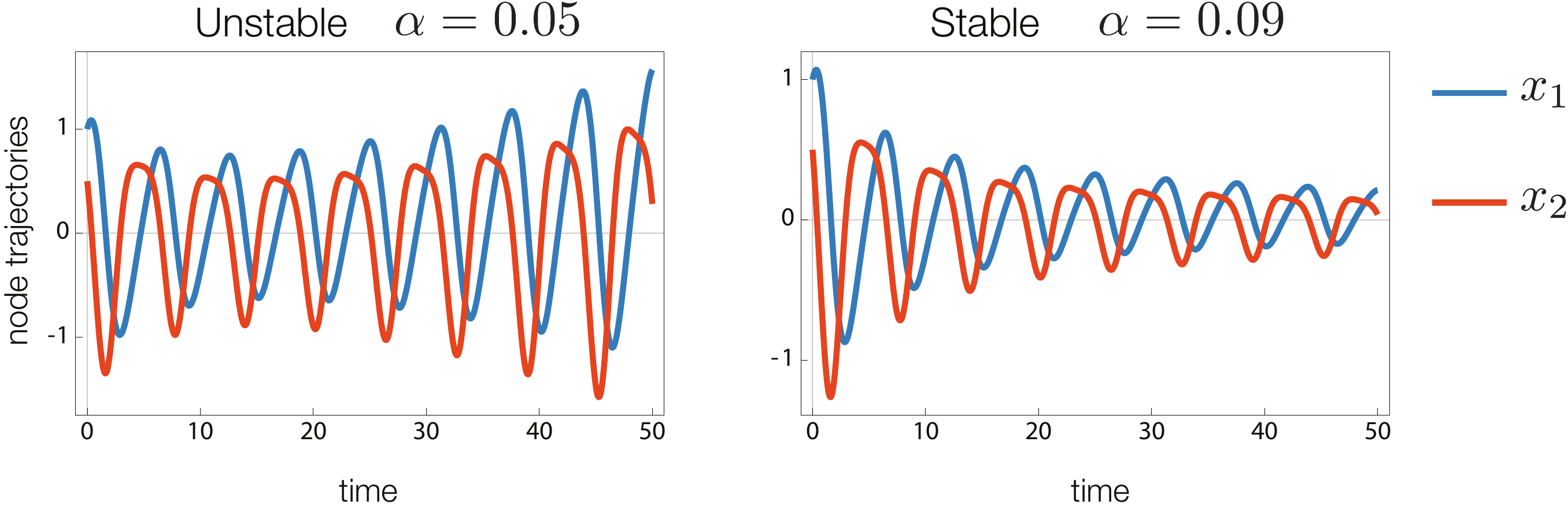}
\caption{{\bf For systems with time-varying asymmetry, a large enough contraction rate is necessary.} Trajectories of the system of Example 3 with initial condition $x(0) = (1, 0.5)^\T$. We observe that the system is unstable  if the contraction rates of the nodes $\alpha \leq 0.05$, but stable if $\alpha \geq 0.09$.
}
\label{fig:unstable}
\end{center}
\end{figure}

By regarding the differential system $\dot {\delta}_x  = A(x(t), t) \delta_x$ as a linear time-varying system and using the recursive proof of Proposition 1, it is straightforward to extend the analysis of the above example to feedback chains and hence to sign-stable systems.
Consequently, time-varying asymmetries do not detriment stability if they are either slow enough or fast enough compared to the system dynamics. 
Biological and man-made systems  often have   time scale separations, so we can exploit them to establish the sign-stability of networked systems \cite{simon1962architecture, del2013contraction}. 
With the help of singular perturbation techniques,  time scale separations in the nodal dynamics can be used to apply the sign-stability criterion  to slow nodes only. 
Importantly, Example 4  implies that a pure qualitative criterion for the stability of  systems with time-varying asymmetries is impossible.

\subsection{Vector nodal dynamics.}

Consider  the negative feedback interconnection between two \emph{modules} with states $x_1 \in \mathbb R^{n_1}$ and $x_2 \in \mathbb R^{n_2}$, corresponding to the following differential system
\begin{equation*}
\begin{split}
\dot \delta_{x_1} = A_{11} \delta_{x_1} + A_{12} \delta_{x_2}, \quad 
\dot \delta_{x_2} = - b_{12} A_{12}^\T \delta_{x_1} + A_{22} \delta_{x_2},
\end{split}
\end{equation*}
with $b_{12}(x,t) >0$ scalar and $A_{ij}(x,t) = \partial f_i(x, t)/\partial x_i \in \mathbb R^{n_i \times n_j}$. Suppose that, when isolated, each module is contracting with metric $M_i$ and that these metrics satisfy the \emph{compatibility} condition $M_1 A_{12} = A_{12} M_2$.
Using the metric $M= \blkdiag \{ b_{12} M_1, M_2\}$ we obtain
$$L(A, M) = \begin{pmatrix} b_{12}[ L(A_{11},M_1) + (\dot b_{12}/b_{12}) M_1] & 0 \\ 0 & L(A_{22}, M_2) \end{pmatrix} $$
and the system remains contracting if  $L(A_{11}, M_1) +  (\dot b_{12}/b_{12}) M_1 \prec 0$ and $L(A_{22}, M_2) \prec 0$, in complete analogy to the scalar case of Section III-B. Consequently, the sign-stability criterion can be straightforwardly extended to modules by extending Proposition 1, first  modifying condition (ii) as follows
$$L(A_{ii}, M_i) + \sum_{j \in \mathcal N_i} \frac{\dot b_{ij}}{b_{ij}} M_i \prec 0, \quad i=1, \cdots, n; $$
and  second requiring the additional \emph{metric compatibility} condition $ M_i A_{ij} = A_{ij} M_j,   \forall j \in \mathcal N_i$.  
Without this  compatibility condition, the conditions that each module needs to satisfy in order to guarantee stability of the whole network are not local anymore (i.e., do not depend on the module's feedback neighbors only) because the block-diagonal structure is lost.

\section{Concluding remarks}

This paper bridges a theoretical gap by extending the sign-stability criterion to nonlinear systems, providing also extensions to consider modules and delayed interconnections. 
%
%
%
%
 In practice, it is rare that large networked systems  are entirely sign-stable. Rather, the nonlinear sign-stability criterion can be used to identify portions of the system (i.e., modules) that are stable by the ``topological'' design of their interconnection network  \cite{Angulo:15}.
%
%
It is also possible to combine the sign-stability criterion with other interconnections that preserve contraction, such as ``centralized'' ones
\cite{tabareau2006notes}.
From an engineering viewpoint, our results allow to recursively build large networks which automatically preserve stability using simple conditions on times scales, delays and  signs of the interconnections.

\ifCLASSOPTIONcaptionsoff
  \newpage
\fi

\bibliographystyle{IEEEtran}
\bibliography{IEEEabrv,ContractionTheory}
%

\end{document}